\documentclass[pra,twocolumn,floatfix,superscriptaddress,longbibliography, nofootinbib]{revtex4-2}

\usepackage[utf8]{inputenc}
\usepackage[left]{lineno}
\usepackage[english]{babel}
\usepackage[colorlinks=true,citecolor=teal,linkcolor=orange,urlcolor=orange]{hyperref}
\usepackage{lipsum}
\usepackage{braket}
\usepackage{xcolor}
\usepackage{ragged2e}

\newcommand*{\Tr}{\operatorname{Tr}}

\newcommand{\norm}[1]{\left\lVert#1\right\rVert}

\providecommand{\myvec}[1]{\ensuremath{\boldsymbol{#1}}}

\providecommand{\xx}{\ensuremath{\myvec{x}}}

\providecommand{\vvartheta}{\ensuremath{\myvec{\vartheta}}}

\providecommand{\calC}{\ensuremath{\mathcal{C}}}
\providecommand{\calD}{\ensuremath{\mathcal{D}}}

\providecommand{\calG}{\ensuremath{\mathcal{G}}}
\providecommand{\calH}{\ensuremath{\mathcal{H}}}

\providecommand{\calO}{\ensuremath{\mathcal{O}}}
\providecommand{\calP}{\ensuremath{\mathcal{P}}}

\providecommand{\calX}{\ensuremath{\mathcal{X}}}
\providecommand{\calY}{\ensuremath{\mathcal{Y}}}

\providecommand{\bbC}{\ensuremath{\mathbb{C}}}

\providecommand{\bbR}{\ensuremath{\mathbb{R}}}

\providecommand{\bbT}{\ensuremath{\mathbb{T}}}

\providecommand{\bbX}{\ensuremath{\mathbb{X}}}

\usepackage{cellspace} 
\usepackage{
    algpseudocode,
    amsmath,
    amssymb,
    amsthm,
    array,
    bbm,
    csquotes,
    dsfont,
    graphicx,
    hyperref,
    lipsum,
    mathtools,
    multirow,
    nicefrac,
    qcircuit,
    slashed,
    times,
    xcolor,
    caption,
    xspace
}
\usepackage[ruled,vlined]{algorithm2e}
\newtheorem{theorem}{Theorem}
\newtheorem{definition}{Definition}

\newtheorem{proposition}[theorem]{Proposition}

\renewcommand{\vec}[1]{\boldsymbol{#1}}

\graphicspath{ {images/} }

\SetKwInput{KwInput}{Input}                
\SetKwInput{KwOutput}{Output}              

\begin{document}

\title{Learning complexity gradually in quantum machine learning models}

\newcommand{\ICFO}{ICFO-Institut de Ciencies Fotoniques, The Barcelona Institute of Science and Technology, 08860 Castelldefels (Barcelona), Spain}
\newcommand{\eure}{Eurecat, Centre Tecnològic de Catalunya, Multimedia Technologies, 08005 Barcelona, Spain}
\newcommand{\FU}{Dahlem Center for Complex Quantum Systems, Freie Universität Berlin, 14195 Berlin, Germany}
\newcommand{\hzb}{Helmholtz-Zentrum Berlin f{\"u}r Materialien und Energie, 14109 Berlin, Germany}
\newcommand{\hhi}{Fraunhofer Heinrich Hertz Institute, 10587 Berlin, Germany}

\author{Erik Armengol}
\affiliation{\ICFO}
\affiliation{\eure}

\author{Franz J. Schreiber}
\affiliation{\FU}

\author{Jens Eisert}
\affiliation{\FU}
\affiliation{\hhi}
\affiliation{\hzb}

\author{Carlos Bravo-Prieto}
\affiliation{\FU}

\begin{abstract}
Quantum machine learning is an emergent field that continues to draw significant interest for its potential to offer improvements over classical algorithms in certain areas. However, training quantum models remains a challenging task, largely because of the difficulty in establishing an effective inductive bias when solving high-dimensional problems. In this work, we propose a training framework that prioritizes informative data points over the entire training set. This approach draws inspiration from classical techniques such as \emph{curriculum learning} and \emph{hard example mining} to introduce an additional inductive bias through the training data itself. By selectively focusing on informative samples, we aim to steer the optimization process toward more favorable regions of the parameter space. This data-centric approach complements existing strategies such as warm-start initialization methods, providing an additional pathway to address performance challenges in quantum machine learning. We provide theoretical insights into the benefits of prioritizing informative data for quantum models, and we validate our methodology with numerical experiments on selected recognition tasks of quantum phases of matter. Our findings indicate that this strategy could be a valuable approach for improving the performance of quantum machine learning models.
\end{abstract}

\maketitle

\section{Introduction}

In light of the impressive success of classical machine learning and the ever-improving experimental means to perform \emph{quantum computing}, the study of the potential and limitations of \emph{quantum machine learning} (QML) have become a topic of widespread interest. Quantum systems are known to perform learning tasks that are beyond the capabilities of classical computers \cite{sweke2021quantum,liu2021rigorous,huang2021information,pirnay2023superpolynomial}. However, there is growing evidence that the training of quantum learning models is hindered by severe problems not present in classical machine learning. Recent investigations have highlighted that the cost landscape of quantum learning models contains many poor minima and vast regions of vanishing gradients, so-called \emph{barren plateaus}, which preclude trainability~\cite{ragone2023unified,mcclean2018barren,larocca2024review,anschuetz2022quantum,NoiseBarrenPlateaus, nemkov2024barren,thanasilp2023subtleties}. 

To address these training challenges and enhance the performance of quantum learning models, it is crucial to introduce a favorable \emph{inductive bias}~\cite{kubler2021inductive, cerezo2022challenges, bowles2023contextuality}. Indeed, barren plateaus arise from random initialization in high dimensional problems, where the cost 
landscape becomes flat, and gradients vanish due to the concentration of measure phenomenon. As such, incorporating a well-advised inductive bias is crucial for achieving reasonable performance. In quantum models, inductive bias has typically been introduced through two main strategies: (1) designing specific quantum circuit architectures that exploit problem symmetries or structures~\cite{larocca2022group, meyer2023exploiting, bowles2023backpropagation, park2024hamiltonian, jerbi2024shadows, schatzki2024theoretical, rodriguez2024training}, and (2) employing smart initialization techniques to set the model parameters near to optimal values, often referred to as \emph{warm-starts}~\cite{zhang2022escaping, grimsley2023adaptive, dborin2022matrix, rudolph2023synergistic, mele2022avoiding, puig2024variational, wang2023trainability, park2024hardware, shi2024avoiding}. Both approaches aim to facilitate the optimization process toward favorable regions of the parameter space, thereby improving trainability and convergence.

However, these approaches primarily focus on modifying the model architecture or parameter initialization, while the role of training data in introducing inductive bias has been largely overlooked. Traditionally, training examples are presented to the learner in random order, a practice common in both classical and quantum machine learning. Nonetheless, there is an increasing awareness in classical machine learning that imposing a specific order on the training data can speed up convergence and lead to better local minima. Typically, this ordering is guided by a \emph{scoring function}, where training starts with examples that have low scores and gradually incorporates data with higher scores as the training progresses. This idea has been explored in different forms, the most prominent examples being \emph{curriculum learning} \cite{bengio2009curriculum}, which was recently applied to quantum models \cite{tran2024quantum}, and \emph{hard example mining} \cite{shrivastava2016training}. These strategies primarily differ in their approach to devising scoring functions. 
Hard example mining focuses on presenting the learner with difficult examples, while curriculum learning, inspired by human learning processes, starts with easier examples and progressively introduces more challenging ones. Despite their differences, both strategies have demonstrated success across various classical machine learning domains, including \emph{natural language processing}, \emph{computer vision} and others \cite{wang2021survey,soviany2022curriculum,tudor2016hard,guo2018curriculumnet,jiang2014easy,platanios2019competence,matiisen2019teacher,zhang2018empirical,shrivastava2016training}. Therefore, adopting an agnostic stance on the design principles is essential, as their quality is ultimately determined by their performance in practice.

Inspired by these powerful classical methods, this work explores an approach to incorporate inductive bias into quantum learning models through the quantum data itself. By carefully selecting, structuring, and presenting the quantum data, we aim to steer the optimization process away from unfavorable regions of the parameter space and guide the model toward better generalization and faster convergence. This data-centric perspective offers a complementary strategy to existing methods focusing on circuit design and parameter initialization. Our contributions include developing a unified and intuitive quantum analog of curriculum learning and hard example mining, tailored to the characteristics of quantum models. Through theoretical analysis and empirical validation, we demonstrate that leveraging the structure and presentation of quantum data can improve the trainability and performance of quantum learning models. This approach not only enhances optimization but also opens new avenues for introducing inductive bias in quantum machine learning. More broadly, our work represents a step toward importing sophisticated methods from classical machine learning into the quantum domain, reflecting the growing maturity of the field.\\

\section{Theory}\label{sec:CL}

In this section, we develop the theoretical underpinning of our framework, starting by laying out a number of preliminaries.


\subsection{Learning setup and training procedure}
We consider a supervised learning scenario involving training data consisting of 
a collection of quantum state vectors $
\{
\ket{\psi_i} 
\in \calH \}$ and their respective labels $y_i \in \calY \subseteq \bbR$. Our quantum learning model, represented as $f_{\vvartheta}:\calH \to \calY$, has trainable parameters $\vvartheta = (\vartheta_1, \vartheta_2, \dots)$. The training set is denoted by $\bbX = \{X_i\}_{i=1}^N = \{\ket{\psi_i}, y_i \}_{i=1}^N$. Depending on the specific learning task, the 
collection of quantum state vectors $
\{\ket{\psi_i}
\}$ may naturally arise or encode classical training data through a particular embedding process. Our goal is to identify a set of parameters $\vvartheta$ that minimizes a given loss function $\ell(\vvartheta; (\ket{\psi}, y))$. 
 
Traditionally, the entire training dataset $\bbX$ is accessible from the start of the training. A common approach is to randomize the order of $\bbX$, split it into mini-batches, and perform gradient-based optimization. Variations of this setup exist, but typically, the order in which training data is presented to the learner is not prioritized. Here, we deviate from this conventional setup by initially limiting the learner to a subset of the training data, gradually increasing access to the full dataset. This framework can be visualized in Fig.~\ref{fig:framework}.

\begin{figure}[t]
    	\centering
     \includegraphics[width=0.48\textwidth]{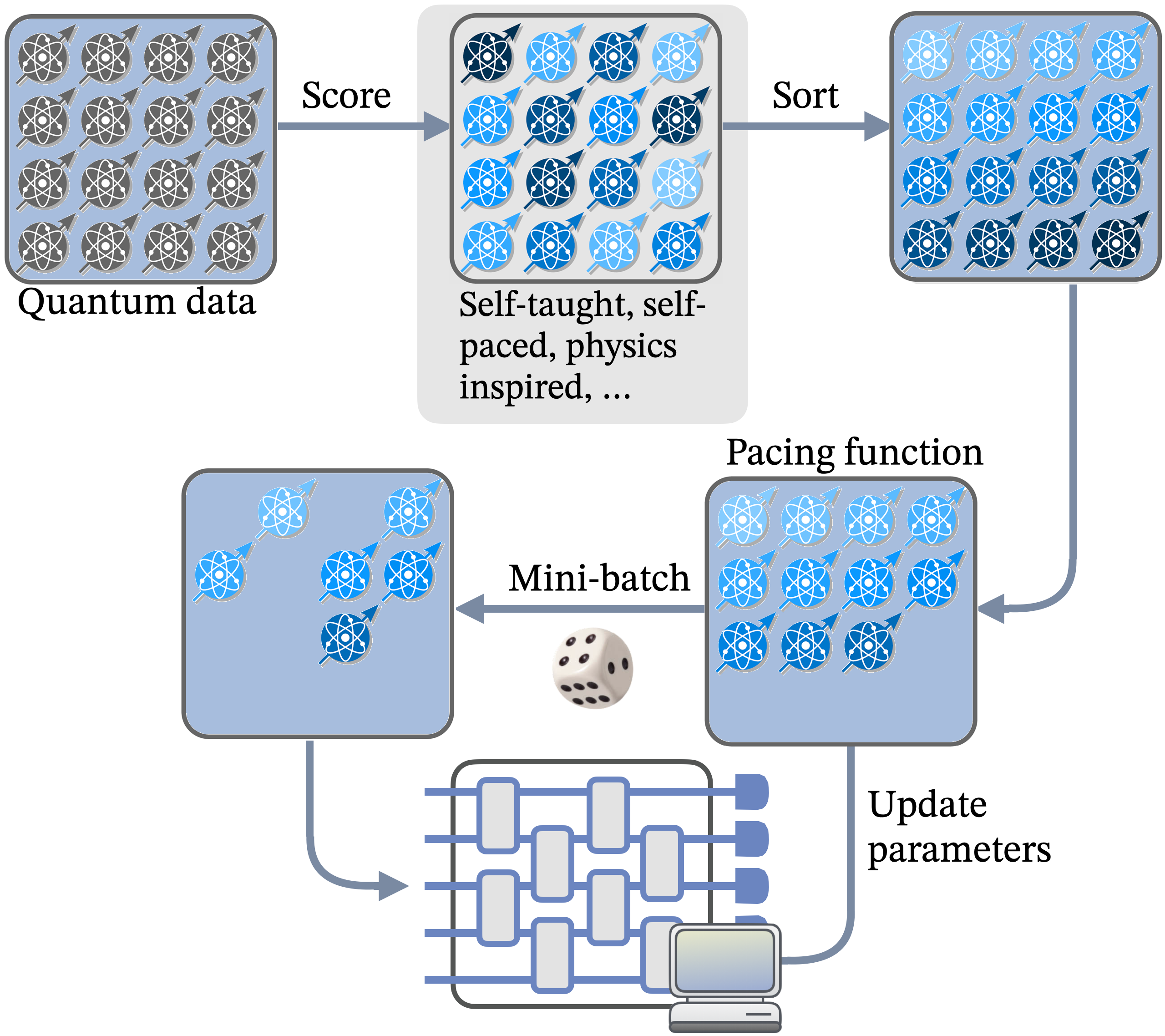}
    	\caption{\justifying
        The training framework begins by scoring quantum data points according to a predefined scoring function (e.g., self-taught, self-paced, or physics-inspired). Data points are sorted by score, and the pacing function determines how much of the dataset is accessible at each training epoch. Mini-batches are then sampled from the available subset, and used to update the model parameters iteratively, enabling a gradual change in learning complexity throughout training.}
    	\label{fig:framework}
\end{figure}

Practically, this is done by introducing a scoring and pacing function. The scoring function determines the order in which data points are added to the subset accessible to the learner by assigning a score $s(\ket{\psi})$ between $0$ and $1$ to each data point $\ket{\psi}$.

\begin{definition}[Scoring function]
Given a dataset $\{\xx_i, y_i\}_{i=1}^N$ with $\xx_i \in \calX$, a scoring function is defined as $s:\calX \to [0,1]$.
\end{definition}

For two data points $\xx_i$ and $\xx_j$, if $s(\xx_i) \leq s(\xx_j)$, then $\xx_j$ should not be available to the learner before $\xx_i$.  
The pacing function essentially dictates the rate at which new data becomes available to the learner by assigning to each epoch a fraction of the training data. Importantly, our framework is also immediately compatible with training on classical data. For example, some methods embed classical data into quantum states, while others, such as data re-uploading~\cite{perez2020data,schuld2021effect}, encode it directly into quantum circuits. Our framework accommodates such approaches, as the scoring function can operate on either the classical data or their corresponding embedding quantum states without conceptual inconsistency.

\begin{definition}[Pacing function]
    Let $\bbT \coloneqq \{1, 2, \dots, T\}$ enumerate the $T$ training epochs. A pacing function is defined as a function $p:\bbT \to [0,1]$.
\end{definition}
Therefore, in epoch $t$, the learner has access to the fraction $p(t)$ of the entire dataset based on the scores of $s(\ket{\psi})$.

\subsection{Scoring and pacing functions}
Scoring functions evaluate the training data, often guided according to notions of difficulty or complexity. While we take an agnostic stance on the idea of scoring functions as difficulty measures, it is noteworthy that those inspired by complexity or difficulty, as previously mentioned, have a successful track record in classical machine learning, often capturing data structures that enable better learning.

We briefly review scoring functions used in practice, distinguishing between \emph{domain-agnostic} and \emph{domain-specific} methods. Domain-agnostic scoring functions, despite their name, are typically tied to the training set and are often learned. This involves initially training a secondary learner on the data and using the resulting loss values as a scoring metric. For example, a learner $f_{\vvartheta}$ may first undergo conventional training, and its hypothesis is subsequently used as a scoring function. Alternative architectures, such as \emph{support vector machines}, can also be adapted in this way. However, this approach can be computationally intensive due to the additional learner.

To mitigate such costs, \emph{self-paced learning} offers a simpler alternative. Here, the loss $\ell_{\vec{\vartheta}}$ of the current hypothesis $f_{\vec{\vartheta}}$ acts as the scoring function ordering the training data of the subsequent epoch. This method preserves the domain-agnostic nature of scoring while avoiding the overhead of training an auxiliary learner.
In contrast, domain-specific scoring functions leverage domain expertise to define metrics, often inspired, as mentioned previously, by complexity measures such as input length in text-based tasks or the number of objects in an image.

Unlike scoring functions, pacing functions dictate the rate at which training examples are introduced. They are generally domain-agnostic and are chosen to be monotonically increasing. Linear pacing functions are the most common baseline, but non-linear alternatives, such as the \emph{root-p} and \emph{geometric progression} functions~\cite{penha2020curriculum, platanios2019competence,wang2021survey}, are also employed. These can be categorized by their growth relative to linear functions: slower-growing functions extend training on high-scoring examples while faster-growing ones prioritize low-scoring data.

\section{Results}
\label{sec:an_results}

In this section, we aim to provide 
 theoretical justification for the use of this learning setup and present theoretical results within the framework proposed.

\subsection{Analytical results on the mitigation of barren plateaus}

In this subsection, we demonstrate how a learning setup defined by a scoring function $s$ and a pacing function $p$ can address challenges during early training epochs and improve convergence rates. We show that carefully selecting and ordering training data can enhance the training dynamics of parameterized quantum models.\\

\begin{proposition}[Mitigation of barren plateaus]\label{prop:barren_plateaus}
Let $\{\ket{\psi_i} \in \calH\}$ be 
a collection of quantum state vectors.
Let $f_{\vvartheta}: \mathcal{H} \rightarrow \mathcal{Y}$ be a parameterized quantum model with parameters $\vvartheta \in \Theta$ and $\calY \subseteq \bbR$ the label space. Consider a loss function $\ell(\vvartheta; (\ket{\psi}, y))$ associated with data points $\{(\ket{\psi_i}, y_i) \in \mathcal{H} \times \mathcal{Y}\}$. Suppose there exists a scoring function $s:\calH \to [0,1]$ such that the expected squared norm of the gradient
    \begin{equation}
        G(z) = \mathbb{E}\left[\norm{\nabla_{\vvartheta} \ell(\vvartheta; (\ket{\psi}, y)}^2 \mid s(\psi) = z\right]
    \end{equation}
    is a non-increasing function of $s$. Then, this learning setup defined by $s$ increases the expected squared gradient magnitude during early training epochs compared to standard training (uniform access to the entire dataset).
\end{proposition}

The proof is given in Appendix~\ref{app:prop_1}. By prioritizing data points with lower scores---corresponding to higher expected squared gradient norms---we mitigate the effects of barren plateaus during the initial stages of training. We later give an explicit example of a physically-inspired scoring function $s$ such that Proposition~\ref{prop:barren_plateaus} holds. 

\subsection{Analytical results on faster convergence}

Next, we show that this learning setup could also improve the convergence of the model.

\begin{proposition}[Faster convergence]\label{prop:convergence}
Let $\{\ket{\psi_i} \}\in \calH$ be 
a collection of
quantum state vectors. Let $f_{\vvartheta}: \mathcal{H} \rightarrow \mathcal{Y}$ be a parameterized quantum model with parameters $\vvartheta \in \Theta$ and $\calY \subseteq \bbR$ the label space. Consider a convex loss function with Lipschitz 
continuous gradients $\ell(\vvartheta; (\ket{\psi}, y))$ associated with data points $\{(\ket{\psi_i}, y_i) \in \mathcal{H} \times \mathcal{Y}\}$. Suppose there exists a scoring function $s:\calH \to [0,1]$ such that the scoring function $s(\ket{\psi})$ prioritizes data points with lower gradient variance at each epoch, and the pacing function $p$ increases monotonically over epochs $t = 1$ to $T$. Then, this learning setup defined by s and p achieves a theoretical upper bound on its expected empirical risk after T epochs that is less than or equal to the bound for standard training (uniform access to the entire dataset).
\end{proposition}
The proof is given in Appendix~\ref{app:prop_2}. This proposition establishes that by prioritizing data points with lower gradient variance, our learning setup achieves a tighter worst-case convergence guarantee. While this does not strictly prove that the expected risk will be lower in every scenario, it provides a strong theoretical motivation for why this data-ordering strategy is beneficial: it bounds the sub-optimality more tightly, suggesting more stable and reliable training dynamics. The assumption of Lipschitz continuity for the gradients assumption is reasonable, given that the empirical risk $R(\vvartheta)$ is generally bounded. The convexity assumption ensures a well-behaved loss landscape, which, while not always realistic, makes the analysis tractable. However, this assumption might be particularly reasonable in scenarios where this learning setup is combined with other warm-start techniques that initialize the parameters near to optimal values. By employing a scoring function that prioritizes such data points and a pacing function that gradually increases the complexity of the training set, the learning process can benefit from reduced variance in gradient estimates, resulting in faster convergence and lower expected risk after a fixed number of epochs.

While these propositions provide a theoretical intuition for the use of hard example mining and curriculum learning in quantum machine learning, empirical validation remains crucial. The theoretical benefits outlined may vary in practice due to factors such as the specific architecture of the quantum model, the characteristics of the data distribution, and the optimization algorithm used. Constructing meaningful scoring and pacing functions remains challenging, and empirical experiments are essential for assessing the practical effectiveness of this learning setup. Such validation can help identify discrepancies between expected and observed performance, allowing adjustments to optimize the approach for specific applications. 
Interestingly, although Propositions~\ref{prop:barren_plateaus} and~\ref{prop:convergence} lay the theoretical ground for what to expect, the numerical performance in certain cases significantly exceeds these initial insights.

\subsection{Learning setup}\label{sec:num_results}
The methodology and underlying principles presented here are generally applicable to quantum machine learning tasks. For specificity, however, we focus on phase recognition. This section presents numerical results demonstrating the application of the presented framework to quantum phase recognition, a popular benchmark task in QML~\cite{caro2022generalization, wu2023quantum, liu2023model, gil2024understanding, umeano2023can, zapletal2024error} due to its importance in the study of condensed matter physics~\cite{Carrasquilla, Sachdev}. We explore its effectiveness on two paradigmatic quantum spin chains: the \emph{generalized cluster model} and the \emph{bond-alternating XXZ model}, both well-studied quantum many-body Hamiltonians.

The \emph{generalized cluster model} is described by the following local Hamiltonian
\begin{align}
        H_{\text{cluster}} =  \sum_{j=1}^n \left(Z_j - j_1 X_j X_{j+1} - j_2 X_{j-1} Z_j X_{j+1}\right)\,,\label{eq:cluster_hamiltonian}
    \end{align}
where $n$ is the number of qubits, $j_1$ and $j_2$ are coupling strengths, and $X_i$ and $Z_i$ are Pauli operators acting on the $i^\text{th}$ qubit. As shown in Refs.~\cite{verresen2017one,PhysRevB.84.165139}, and illustrated in Fig.~\ref{fig:diagram}(a), the ground-state phase diagram exhibits four distinct regimes: (I) a symmetry-protected topological phase, (II) a trivial phase, (III) an antiferromagnetic phase, and (IV) a ferromagnetic phase. For specific choices of parameters, this Hamiltonian is the standard cluster model, the conventional cluster state being the ground state.

The \emph{bond-alternating XXZ Hamiltonian} is slightly more complicated, although it remains geometrically local, and reads
\begin{equation}
\begin{split}
        H_{\text{XXZ}} =  j_1\sum_{j=1}^{n/2} \left(X_{2j-1} X_{2j} + Y_{2j-1} Y_{2j} + \delta Z_{2j-1} Z_{2j}\right)\\
        + j_2\sum_{j=1}^{n/2-1} \left(X_{2j} X_{2j+1} + Y_{2j} Y_{2j+1} + \delta Z_{2j} Z_{2j+1}\right)\,,\label{eq:XXZ_hamiltonian}
    \end{split}
    \end{equation}
where $j_1$ and $j_2$ are coupling strengths, $X_i$, $Y_i$ and $Z_i$ are Pauli operators acting on the $i^\text{th}$ qubit, and $\delta$ is the spin anisotropy. As depicted in Fig.~\ref{fig:diagram}(b), this model Hamiltonian hosts three distinct phases that can be detected by the many-body topological invariant $\mathcal{Z_R}$~\cite{elben2020many}: (I) a topological phase, (II) a trivial phase, and (III) a antiferromagnetic phase.

Our learning tasks involve developing a classification model that can identify the quantum phase given the ground state of a system described by either the \emph{generalized cluster Hamiltonian} or the \emph{bond-alternating XXZ Hamiltonian}. We achieve this by generating a training set, denoted by $S=\{(\lvert\psi_i\rangle,\vec{y_i})\}_{i=1}^N$, where $N$ represents the number of training data points. Each element in $S$ is a pair: $\lvert\psi_i\rangle$ corresponds to the ground state vector of the Hamiltonian $(H_{\text{cluster}}$ or $H_{\text{XXZ}})$ for a specific choice of coupling constants $(j_1$ and $j_2$ or $j_1/j_2$ and $\delta)$. The ground states are drawn from a distribution $\calD$ which corresponds to sampling these coupling constants uniformly at random, ensuring a balanced representation of all classes. The second element, $\vec{y_i}$, represents the corresponding one-hot encoded phase label. In this scheme, each label is a vector with a length equal to the number of classes, with only one element set to 1, indicating the corresponding phase.

We leverage the \emph{quantum convolutional neural network} (QCNN) architecture introduced in Ref.~\cite{cong2019quantum} for the phase classification tasks. Inspired by classical convolutional neural networks, QCNNs are expected to excel at identifying spatial or temporal patterns, making them well-suited for this purpose. A key feature of QCNNs is their alternating structure of convolutional and pooling layers. Convolutional layers apply parameterized, translation-invariant unitary operations to neighboring qubits, acting like filters between feature maps across layers. Pooling layers then reduce the dimensionality of the quantum state while preserving the relevant information. This is achieved by measuring a subset of qubits and applying single-qubit unitaries based on the measurement outcomes. Each pooling layer consistently reduces the number of qubits by a constant factor, leading to quantum circuits with logarithmic depth compared to the initial system size.

\begin{figure}
    	\centering
    	\includegraphics[width=0.48\textwidth]{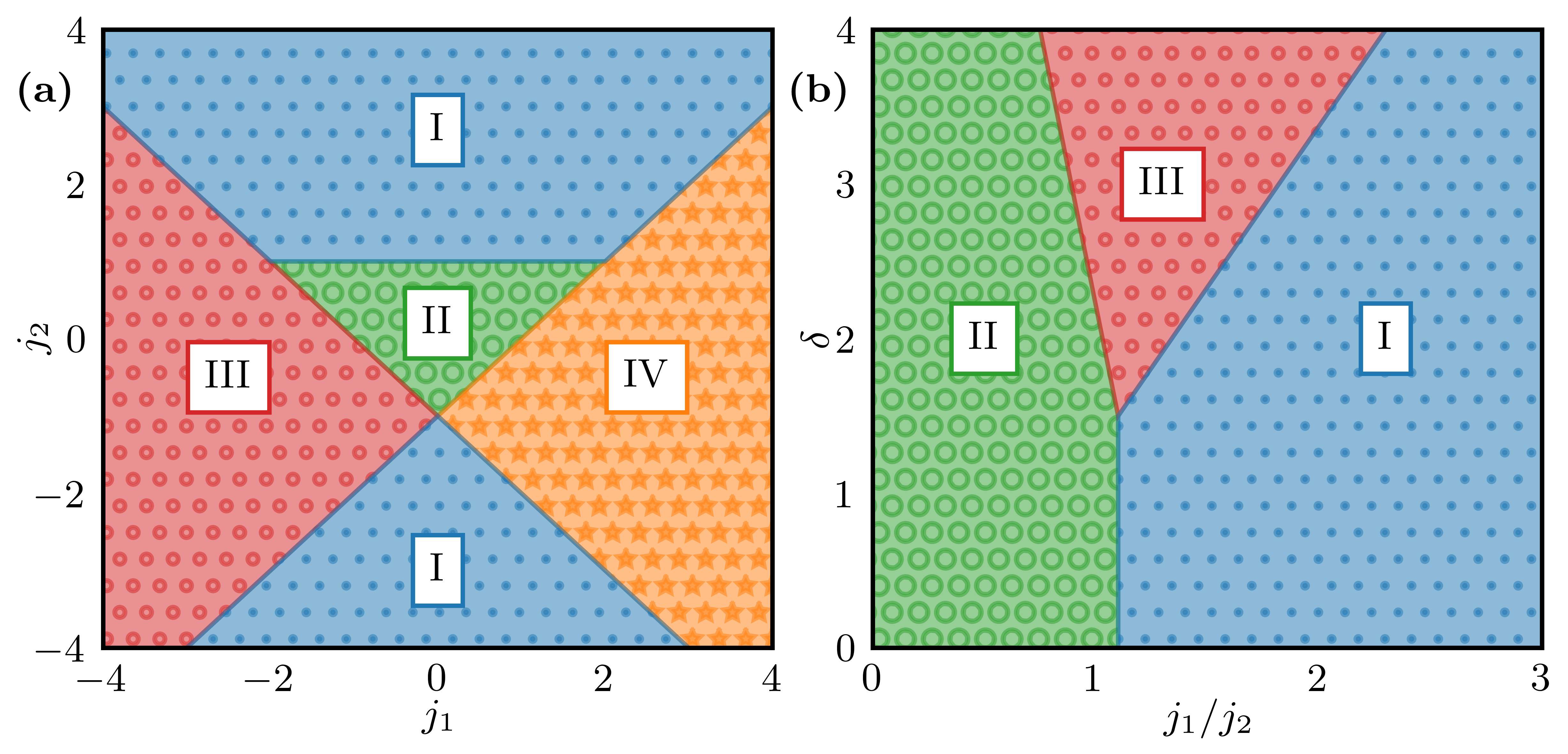}
    	\caption{\justifying
        Ground-state phase diagrams of \textbf{(a)} the \emph{generalized cluster Hamiltonian} in Eq.~\eqref{eq:cluster_hamiltonian} and \textbf{(b)} the \emph{bond-alternating XXZ Hamiltonian} in Eq.~\eqref{eq:XXZ_hamiltonian}, exhibiting (I) symmetry-protected topological, (II) trivial, (III) antiferromagnetic, and (IV) ferromagnetic phases.}
    	\label{fig:diagram}
\end{figure}

We can describe the action of the QCNN as a quantum channel $\calC_{\vec{\vartheta}}$ parameterized by $\vec{\vartheta}$. This channel transforms an input state $\rho_{\text{in}}$ into an output state $\rho_{\text{out}}$ via the mapping $\rho_{\text{out}} = \calC_{\vec{\vartheta}}\left[\rho_{\text{in}}\right]$. Following this, a task-specific Hermitian operator is measured on the output state $\rho_{\text{out}}$ to obtain an expectation value. In our numerical implementations, the QCNN maps an $8$-qubit input state vector $\lvert\psi\rangle$ into a $2$-qubit output state. The labeling function for the output state utilizes the probabilities obtained by measuring the state in the computational basis,
giving rise to outcomes $(0,0)$, $(0,1)$, $(1,0)$, $(1,1)$.
Specifically, the predicted label $\vec{\hat{y}}$ corresponds to the vector of the probabilities as defined by
   \begin{align}
        \lvert\psi\rangle \mapsto (p_j)_{j\in\{0,1\}^2} \mapsto \vec{\hat{y}} \coloneqq (p_{0,0}, p_{0,1}, p_{1,0}, p_{1,1})\,.
    \end{align}
For each experimental run, data is generated from the corresponding 
distribution $\calD$. During training, we employ the mean squared error as loss function
\begin{align}
    \ell\left(\vec{\vartheta};\left(\lvert\psi\rangle, \vec{y}\right) \right) &\coloneqq \sum_{j=1}^ M \left(\langle j \rvert\left(\calC_{\vec{\vartheta}}\left[\lvert\psi\rangle\!\langle\psi\rvert\right]\right)\lvert j \rangle - y_j \right)^2\,,
    \end{align}
where $M$ is the number of classes. Then, given a training set of size $N$ drawn from the distribution $\calD$, we sample a mini-batch of size $L$, and we minimize the empirical risk
\begin{align}
    \hat{R}_S(\vec{\vartheta})  &= \frac{1}{L} \sum_{i=1}^L \ell\left(\vec{\vartheta};\left(\lvert\psi_i\rangle, \vec{y}_i\right) \right).
\end{align}

Let us detail here the experimental setup used throughout the following sections. We generate a training set $S$ containing 50 ground states drawn from the distribution $\calD$, and a test set $T$ containing 1000 samples from the same distribution. While the ratio of training to test data is unconventional compared to classical machine learning practices, it is intentionally chosen for this quantum context. Because obtaining ground-state data for quantum many-body systems is resource-intensive, demonstrating robust generalization from a minimal number of training examples is a key benchmark for QML models~\cite{caro2022generalization, gil2024understanding}. Furthermore, restricting the training set size ensures the learning task is sufficiently demanding to clearly isolate and highlight the performance benefits of our data-ordering strategies. The large test set is maintained to provide high statistical confidence in the evaluated generalization performance. 

To ensure a balanced representation of all classes, these states are obtained by sampling the Hamiltonian's coupling constants an equal number of times from within each distinct phase region. We refer to the standard approach of uniformly sampling mini-batches of size $L$ from the training set $S$ as the \textit{standard} method. In this work, we use a fixed mini-batch size of $L=10$. Each learning framework is evaluated over ten independent runs, using newly generated data and random initializations for each run. Next, we define a scoring function $s : S \rightarrow \mathbb{R}$ to measure the difficulty of each training example $(\lvert\psi_i\rangle,\vec{y}_i)$. This function assigns a higher score to harder examples, e.g., $s((\lvert\psi_i\rangle,\vec{y}_i)) > s((\lvert\psi_j\rangle,\vec{y}_i))$. The scoring function $s$ incorporates prior knowledge and thus embodies some of the inductive bias of the model. We additionally define a pacing function to determine how the size of the subsets from which mini-batches are uniformly sampled is adjusted throughout training. Details regarding the implementation and the specific pacing functions employed are provided in Ref.~\cite{github_repository} and Appendix~\ref{app:numerics}.

For the quantum circuit simulations, we leverage the {\tt PennyLane}~\cite{bergholm2018pennylane} software library running on classical hardware. We utilize the ADAM optimizer~\cite{kingma2014adam} for parameter updates during training, which is a stochastic gradient descent method that dynamically adjusts learning rates.\\

\subsection{Self-taught learning}\label{sec:self-taught_CL}
In this first experiment, we order the quantum data based on the loss of a pre-trained QCNN model with identical architecture and trained without ordering. This is usually referred to as \emph{self-taught} learning. At the beginning of the training, we order the training set according to the loss of each data point. We then gradually increase the number of data points used for training throughout the epochs, following a monotonically increasing pacing function. We consider four orderings: (1) \textit{Standard}, which is used as a baseline and corresponds to the standard approach with no specific ordering nor pacing function, all data points are treated equally from the start; (2) \textit{Random}, where data points are randomly ordered, and serves as a control that should behave similar to the Standard approach; (3) \textit{Easy}, which prioritizes data points with lower loss values (easier examples); and (4) \textit{Hard}, which prioritizes data points with higher loss values (harder examples).

\begin{figure}[t!]
    \centering

    \includegraphics[width=0.48\textwidth]{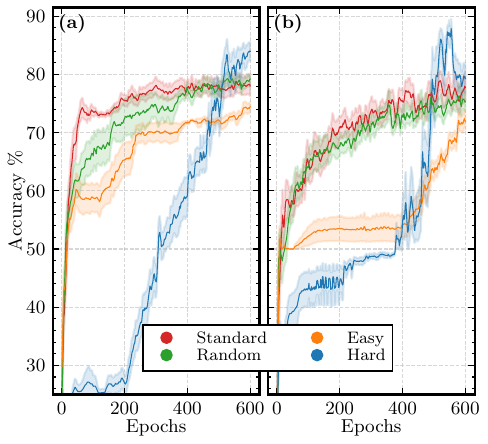}
    \caption{\label{fig:self-taught}
    \justifying
    Average accuracy on the test set as a function of the training epochs for the quantum phase recognition task for the \textbf{(a)} \emph{generalized cluster model} and \textbf{(b)} the \emph{bond-alternating XXZ model}, using a self-taught strategy. The shaded areas represent the standard error of the mean over ten runs with different random initialization and training data.}

    \vspace{1.5em}

    \begin{minipage}{\columnwidth}
        \centering
        \renewcommand{\arraystretch}{1.2}

        \begin{tabular}{|Sc|Sc|Sc||Sc|Sc|}
            \cline{2-5}
            \multicolumn{1}{c|}{} & \multicolumn{2}{c||}{\textbf{Generalized cluster}} & \multicolumn{2}{c|}{\textbf{Bond-alternating XXZ}} \\
            \cline{2-5}
            \multicolumn{1}{c|}{} & \multicolumn{1}{c|}{\,\,Training $(\%)$\,\,} & \multicolumn{1}{c||}{\,\,Test $(\%)$\,\,} & \multicolumn{1}{c|}{\,\,Training $(\%)$\,\,} & \multicolumn{1}{c|}{\,\,Test $(\%)$\,\,} \\
            \hline
            \,\textit{Standard}\, & 79.4 & 78.2 & 81.2 & 78.6 \\
            \hline
            \,\textit{Random}\,   & 80.6 & 79.2 & 78.8 & 75.8 \\
            \hline
            \,\textit{Easy}\,     & 75.6 & 74.4 & 73.4 & 72.4 \\
            \hline
            \,\textit{Hard}\,     & \textbf{86.6} & \textbf{83.6} & \textbf{88.0} & \textbf{86.5} \\
            \hline
        \end{tabular}

        \captionsetup{justification=justified,singlelinecheck=false}
        \captionof{table}{\label{table:ST_CL}
        \justifying
        Average best accuracy over ten runs for the quantum phase recognition task on the \emph{generalized cluster model} and the \emph{bond-alternating XXZ model}, using a self-taught strategy. The best-performing method prioritizes hard examples early in the training.}
    \end{minipage}
\end{figure}

Figure~\ref{fig:self-taught} shows the average performance over ten runs of these strategies on the \emph{generalized cluster} and the \emph{bond-alternating XXZ models}. For both spin models, the Standard and Random strategies exhibit comparable performance, achieving the best accuracy nearing $80\%$. This similarity arises because Random introduces random ordering, mirroring the uniform treatment of data in the Standard approach. Interestingly, the Easy strategy, which prioritizes easier examples, performs worse than both Standard and Random, converging to a lower accuracy and suggesting that early exposure to easier examples might not be optimal. Conversely, the Hard strategy, which focuses on harder examples, despite initially slower convergence, ultimately yields the best performance, approaching $90\%$ accuracy (see Table~\ref{table:ST_CL}). Notably, for the \emph{bond-alternating XXZ model} shown in Fig.~\ref{fig:self-taught}(b), the highest accuracy is reached before the end of the training and then it decreases in the final optimization steps due to the incorporation of the easiest examples in the training batch. This suggests that exposure to non-informative data points might even be detrimental for the performance of the quantum model, and highlights the potential benefit of focusing on complex examples early in the training.

\subsection{Self-paced learning}\label{sec:self-paced_CL}

\textit{Self-paced} learning dynamically prioritizes training data based on the performance of the current model, unlike the self-taught strategy which scores each point based on the loss with respect to the target hypothesis.
In this approach, we rank the training set by the current loss of each data point at each epoch. We refer the reader to Appendix~\ref{app:self-paced} for an extended discussion on the complexity and overhead of this strategy. We explore four strategies within this framework, using a monotonically increasing pacing function to gradually expose more data points as training progresses: (1) \textit{Standard}, (2) \textit{Easy}, and (3) \textit{Hard}, which we discussed previously, and (4) \textit{Hardest}, which employs a constant pacing function, consistently training on the ten most difficult examples at each epoch. This means that, although the pacing function is constant, the most difficult examples can vary at each epoch depending on the performance of the model. 

The performance of these strategies is illustrated in Fig.~\ref{fig:self-paced}, showing the average accuracy on the training and test sets for the quantum phase recognition task across the two spin models. Standard and Easy strategies achieve similar performance, with a final accuracy nearing $80\%$. This indicates that a straightforward or Easy approach does not significantly boost the model's learning capability. The Hard strategy initially shows promising performance but regresses as more data points are introduced due to the increasing pacing function. Despite this regression, it demonstrates the potential benefit of focusing on complex examples early in the training. In striking contrast, the Hardest strategy quickly converges to near-perfect accuracy on the training set and achieves over $90\%$ accuracy on the test set (see Table~\ref{table:SP_CL}). By consistently training on the hardest examples, the model develops a robust understanding of the most challenging aspects of the training data, such as the boundaries of the different quantum phases---as we further explore in Sec~\ref{sec:accuracy}---hence significantly enhancing its performance.

\begin{figure}[t!]
             \centering             \includegraphics[width=0.48\textwidth]{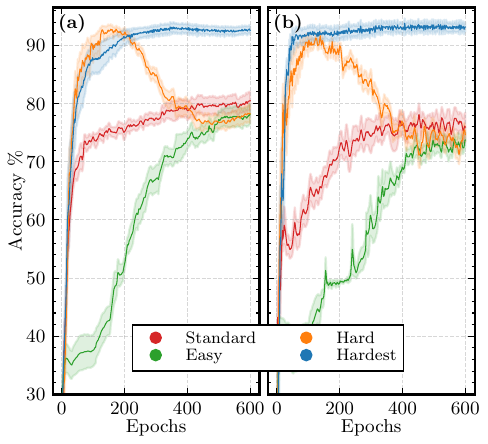}
             \caption{\label{fig:self-paced}
             \justifying
             Average accuracy on the test set as a function of the training epochs for the quantum phase recognition task for the \textbf{(a)} \emph{generalized cluster model} and \textbf{(b)} the \emph{bond-alternating XXZ model}, using a self-paced strategy. The shaded areas represent the standard error of the mean over ten runs with different random initialization and training data.
                }
    \end{figure}
\begin{table}[t!]
\centering
\renewcommand{\arraystretch}{1.2} 

\begin{tabular}{|Sc|Sc|Sc||Sc|Sc|} \cline{2-5}
\multicolumn{1}{c|}{} & \multicolumn{2}{c||}{\textbf{Generalized cluster}} & \multicolumn{2}{c|}{\textbf{Bond-alternating XXZ}} \\ \cline{2-5}
\multicolumn{1}{c|}{} & \multicolumn{1}{c|}{\,\,Training $(\%)$\,\,} & \multicolumn{1}{c||}{\,\,Test $(\%)$\,\,} & \multicolumn{1}{c|}{\,\,Training $(\%)$\,\,} & \multicolumn{1}{c|}{\,\,Test $(\%)$\,\,} \\ \hline
\,\textit{Standard}\, & 83.2 & 77.4 & 82.2 & 77.4 \\ \hline
\,\textit{Easy}\, & 80.8 & 78.1 & 76.4 & 73.5\\ \hline
\,\textit{Hard}\, & 98.6 & \textbf{92.6} & 95.0 & 91.6\\ \hline
\,\textit{Hardest}\, & \textbf{99.8} & \textbf{92.6} & \textbf{97.2} & \textbf{93.5}\\
\hline\end{tabular}
\caption{\justifying
Average best accuracy over ten runs for the quantum phase recognition task on the \emph{generalized cluster model} and the \emph{bond-alternating XXZ model}, using a self-paced strategy. The best-performing method prioritizes the hardest examples in each training epoch.}
 \label{table:SP_CL}
\end{table}

\subsection{Physics-inspired learning}\label{sec:physics-inspired_CL}

A key difference between classical and quantum machine learning is, once again, the presence of barren plateaus, where the variation of the loss function over the parameter landscape is exponentially suppressed with the system size \cite{mcclean2018barren}. At its core, this is a concentration of measure effect in high dimensions and can be captured in several different ways. One effective approach is through the \emph{dynamical Lie algebra} (DLA) of the circuit ansatz \cite{ragone2023unified,AguilarLie}.
The DLA $\mathfrak{g}$ is defined as the Lie closure over its set of generators, 
\begin{align}
    \mathfrak{g} \coloneqq \left<i\calG \right>_{\text{Lie}}\, ,
\end{align}
where the generators are the generators of the parametrized gates of the circuit, $\calG \coloneqq \{H_1, H_2, \dots, H_K\}$. The DLA $\mathfrak{g}$ is the vector space (over $\bbR$) that is obtained by taking the span over 
all the nested commutators that can be formed from the elements of $i\calG$. Assuming that the quantum neural network approximates a unitary 2-design 
\cite{gross_evenly_2007} over the unitary group generated by the DLA $\mathfrak{g}$,
for the training state $\rho$, the variation of the loss is then given by 
\begin{align} \label{eqn:analytical results: variation of loss formula}
    \mathrm{Var}(\ell(\vvartheta;\rho, O)) = \frac{\calP_{\mathfrak{g}}(\rho) \calP_{\mathfrak{g}}(O)}{\mathrm{dim}(\mathfrak{g})} \, ,
\end{align}
where $P_{\mathfrak{g}}$ denotes the $\mathfrak{g}$-purity, defined for Hermitian operators as
\begin{align} \label{eqn:g_purity}
    \calP_{\mathfrak{g}}(H) \coloneqq \sum_{l=1}^{\mathrm{dim}(\mathfrak{g})} \left| \Tr(B_j^\dagger H) \right|^2 \, .
\end{align}
Here, $\{B_j\}_{j=1}^{\mathrm{dim}(\mathfrak{g})}$ is a Hilbert-Schmidt orthonormal basis of the vector space $\mathfrak{g}_{\bbC} \coloneqq \mathrm{span}(\mathfrak{g})$. From Eq.~(\ref{eqn:analytical results: variation of loss formula}) we observe that barren plateaus occur if either $1/\dim(\mathfrak{g})=\calO(b^{-n})$, $\calP_{\mathfrak{g}}(O)=\calO(b^{-n})$ or $\calP_{\mathfrak{g}}(\rho) = \calO(b^{-n})$ with $b>1$. 

The only term in Eq.~(\ref{eqn:analytical results: variation of loss formula}) that depends on the training data is the $\mathfrak{g}$-purity given in Eq.~(\ref{eqn:g_purity}). Favorable scaling of this quantity is a necessary condition to avoid an exponentially vanishing variance in Eq.~(\ref{eqn:analytical results: variation of loss formula}) and, consequently, the onset of barren plateaus. Since one of the key statistical assumptions of barren plateaus is a random initialization of the trainable parameters, their appearance is most accurate at the beginning of training.
Particularly during this initial training phase, one can expect that training data with higher $\mathfrak{g}$-purity enhances the trainability of the model, as these points are associated, on average, with larger gradients in the asymptotic regime. This, however, does not necessarily hold for finite system sizes, as barren plateaus are an intrinsically asymptotic concept.
Building on this intuition, we propose the following domain-specific scoring function 
\begin{align}
    s(\rho) = 1-\calP_{\mathfrak{g}}(\rho)
\end{align}
where lower scores correspond to training examples with high $\mathfrak{g}$-purity, which are prioritized at the beginning of training. Notice that choosing $s$ this way ensures that Proposition~\ref{prop:barren_plateaus} holds. Accordingly, as training progresses and the model converges to more favorable regions of the cost landscape, the statistical assumptions underpinning barren plateaus become less relevant. As this occurs, we gradually incorporate more training data with higher scores, which correspond to lower $\mathfrak{g}$-purity. Ideally, this way, the model can still learn from the complete training set without low gradient data points impeding the training progress.

\begin{figure}[t!]
             \centering
             \includegraphics[width=0.48\textwidth]{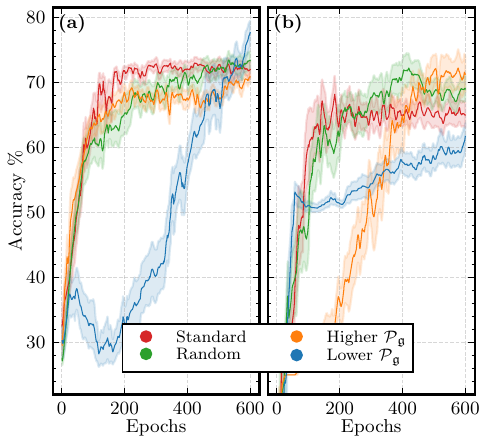}
             \caption{\label{fig:phys-ins}
             \justifying
             Average accuracy on the test set as a function of the training epochs for the quantum phase recognition task for the \textbf{(a)} \emph{generalized cluster model} and \textbf{(b)} the 
             \emph{bond-alternating XXZ model}, using a physics-inspired strategy ($P_{\mathfrak{g}}$ of the training quantum states). The shaded areas represent the standard error of the mean over ten runs with different random initialization and training data.
                }
    \end{figure}
\begin{table}[t!]
\centering
\renewcommand{\arraystretch}{1.2} 

\begin{tabular}{|Sc|Sc|Sc||Sc|Sc|} \cline{2-5}
\multicolumn{1}{c|}{} & \multicolumn{2}{c||}{\textbf{Generalized cluster}} & \multicolumn{2}{c|}{\textbf{Bond-alternating XXZ}} \\ \cline{2-5}
\multicolumn{1}{c|}{} & \multicolumn{1}{c|}{\,\,Training $(\%)$\,\,} & \multicolumn{1}{c||}{\,\,Test $(\%)$\,\,} & \multicolumn{1}{c|}{\,\,Training $(\%)$\,\,} & \multicolumn{1}{c|}{\,\,Test $(\%)$\,\,} \\ \hline
\,\textit{Standard}\, & 71.6 & 72.5 & 70.0 & 67.9 \\ \hline
\,\textit{Random}\, & 72.8 & 72.9 & \textbf{73.8} & \textbf{71.9}\\ \hline
\,\textit{Higher $\calP_{\mathfrak{g}}$}\, & 68.6 & 70.9 & 73.4 & 71.6\\ \hline
\,\textit{Lower $\calP_{\mathfrak{g}}$}\, & \textbf{78.6} & \textbf{77.2} & 62.4 & 60.8\\
\hline\end{tabular}
\caption{\justifying
Average best accuracy over ten runs for the quantum phase recognition task on the \emph{generalized cluster model} and the 
\emph{bond-alternating XXZ model}, using a physics-inspired strategy.}
 \label{table:PI_CL}
\end{table}

A potential limitation of using $\mathfrak{g}$-purity as a score function for larger system sizes is the computational difficulty associated with calculating $\calP_{\mathfrak{g}}$, which generally scales exponentially with the number of qubits. To address this challenge, it is worthwhile to explore approximation schemes for computing $\calP_{\mathfrak{g}}$, which could have broader implications beyond our specific context, particularly for understanding when quantum systems can be efficiently simulated using Lie-algebraic techniques~\cite{goh2023liealgebraicclassicalsimulationsvariational}.
 To this end, we note that there are two sources of computational hardness. First, the Lie algebra itself might be exponentially large, requiring an exponentially sized basis of $\mathfrak{g}_{\bbC}$. In such cases, training would likely be infeasible anyway, as then $1/\dim(\mathfrak{g}) = \calO(b^{-n})$. Second, even if the dimension of the Lie algebra grows only polynomially in the system size, that does not guarantee an efficient method for finding a basis for the space $\mathfrak{g}_{\bbC}$. Computing the basis elements $\{B_j\}_{j=1}^{\mathrm{dim}(\mathfrak{g})}$ requires computing nested commutators of the generators in $i\calG$ and verifying the linear independence the resulting operators (c.f.\ 
 Algorithm~1 in Ref.~\cite{Larocca_2022}). 
 While we have an efficient classical description of the generating set, this does not ensure that the required operations (nested commutators and linear independence checks) can be performed efficiently. Instead, we may need to represent the generators as elements of $\bbC^{2^n}$ to carry out these computations, making the determination of even a single basis element $B_j$ quickly infeasible. There are, however, special cases in which efficient procedures for computing $\{B_j\}_{j=1}^{\mathrm{dim}(\mathfrak{g})}$ are known \cite{goh2023liealgebraicclassicalsimulationsvariational,ragone2023unified}.

To avoid computational bottlenecks in this section, we simplified the QCNN while preserving its structure to a parameterized matchgate circuit composed of fermionic Gaussian unitaries. Specifically, the circuit generators are $\calG = \{Z_i\}^8_{i=1} \cup \{X_i X_{i+1}\}^7_{i=1}$, producing a DLA $\mathfrak{g}$ that remains computationally feasible. This simplification results in a QCNN with fewer variational parameters, leading to reduced performance across all strategies. Nonetheless, this trade-off is acceptable because our goal is not to optimize performance but rather to compare differences between frameworks (e.g., Standard versus alternative methods). We consider four strategies: (1) \textit{Standard} and (2) \textit{Random}, previously discussed, and (3) \textit{Higher $\calP_{\mathfrak{g}}$} and (4)  \textit{Lower $\calP_{\mathfrak{g}}$}, which prioritize higher and lower $\calP_{\mathfrak{g}}$ states early in the training, respectively.

Figure~\ref{fig:phys-ins} shows the accuracy performance of these strategies on the training and test sets for our quantum phase recognition tasks. We observe different behaviors for each spin model under study. For the \emph{generalized cluster model}, the Standard and Random strategies perform similar, as expected. Interestingly, the Higher $\calP_{\mathfrak{g}}$ also performs similarly, though with slightly worse results. Conversely, the Lower $\calP_{\mathfrak{g}}$ strategy shows a gradual improvement initially, eventually achieving the highest performance on both test and training sets (see also Table~\ref{table:PI_CL}).

In contrast, for the \emph{bond-alternating XXZ model}, while the Standard and Random strategies exhibit similar performance to those observed in the \emph{generalized cluster model}, the $\calP_{\mathfrak{g}}$-based strategies behave differently. In this case, the final accuracy of the Higher $\calP_{\mathfrak{g}}$ strategy surpasses that of the Lower $\calP_{\mathfrak{g}}$ strategy, which significantly underperforms compared to the other strategies. Notably, despite the Higher $\calP_{\mathfrak{g}}$ strategy achieving the best final accuracy, the highest accuracy during the training is attained by the Random strategy, indicating that this specific implementation does not benefit from $\calP_{\mathfrak{g}}$ ordering for this particular task.

These results suggest that ordering the quantum data according to $\calP_{\mathfrak{g}}$ does not consistently yield major improvements in learning performance in this particular setup. While ordering based on $\calP_{\mathfrak{g}}$ guarantees larger gradients on average at the early stages of the training, it does not necessarily lead to better learning outcomes.
However, it is evident that $\calP_{\mathfrak{g}}$ captures meaningful properties of the problem, as the performance of the model varies significantly with different orderings. Further research could be beneficial to understand the impact of $\calP_{\mathfrak{g}}$ on training and to develop a scoring function that may leverage its properties in a non-trivial manner.


\subsection{Large-scale implementations}\label{sec:larger_numerics}

To assess the scalability of our proposed framework, we extend our numerical experiments to larger system sizes. In this section, we investigate the performance of the most promising method identified in our earlier experiments (the self-paced learning strategy) on the generalized cluster model for systems of 16 and 32 qubits.

\begin{figure}[t!]
    \centering
             \includegraphics[width=0.48\textwidth]{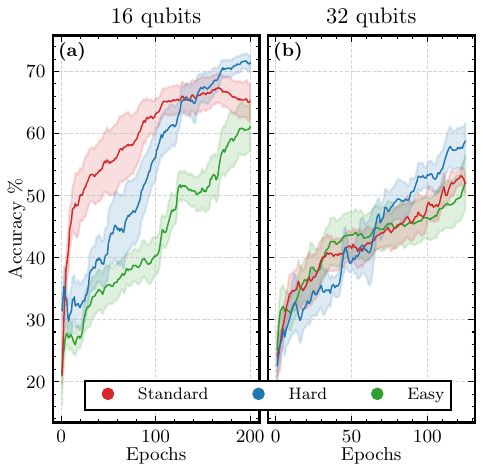}
             \caption{\label{fig:large_CL}
             \justifying
             Average accuracy on the test set as a function of the training epochs for the quantum phase recognition task for the \emph{generalized cluster model} on \textbf{(a)} 16 qubits  and \textbf{(b)} 32 qubits, using a self-paced strategy. The shaded areas represent the standard error of the mean over six runs with different random initialization and training data.
                }

    \vspace{1.5em}

    \begin{minipage}{\columnwidth}
        \centering
        \renewcommand{\arraystretch}{1.2}

        \begin{tabular}{|Sc|Sc|Sc||Sc|Sc|} \cline{2-5}
\multicolumn{1}{c|}{} & \multicolumn{2}{c||}{\textbf{16 qubits}} & \multicolumn{2}{c|}{\textbf{32 qubits}} \\ \cline{2-5}
\multicolumn{1}{c|}{} & \multicolumn{1}{c|}{\,\,Training $(\%)$\,\,} & \multicolumn{1}{c||}{\,\,Test $(\%)$\,\,} & \multicolumn{1}{c|}{\,\,Training $(\%)$\,\,} & \multicolumn{1}{c|}{\,\,Test $(\%)$\,\,} \\ \hline
\,\textit{Standard}\, & 71.7 & 67.4 & 58.7 & 53.2 \\ \hline
\,\textit{Easy}\, & 67.0 & 61.1 & 58.3 & 52.1\\ \hline
\,\textit{Hard}\, & \textbf{73.3} & \textbf{71.7} & \textbf{63.3} & \textbf{58.7}\\
\hline\end{tabular}

        \captionsetup{justification=justified,singlelinecheck=false}
        \captionof{table}{\label{table:large_CL}
        \justifying
        Average best accuracy over six runs for the quantum phase recognition task on the \emph{generalized cluster model} using a self-paced strategy, for 16- and 32-qubit QCNNs.}
    \end{minipage}
\end{figure}

\begin{figure*}[t]
    	\centering
    	\includegraphics[width=0.45\textwidth]{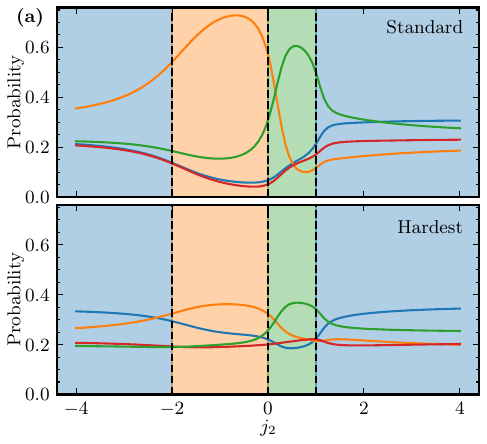}
     \includegraphics[width=0.441\textwidth]{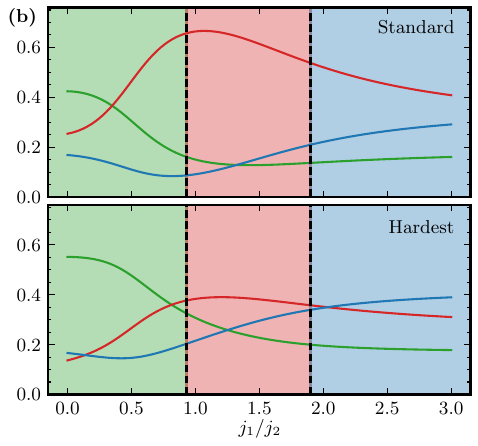}
    	\caption{\justifying
        \textbf{(a)} The \emph{generalized cluster Hamiltonian} and \textbf{(b)} \emph{bond-alternating XXZ Hamiltonian} for Standard (top) and Hardest (bottom) strategies. The background color of each panel indicates the true quantum phase of the system. A correct classification occurs when the highest probability color matches the background color of the panel. The self-paced Hardest strategy demonstrates superior accuracy both within each phase region and at the phase boundaries.}
    	\label{fig:probabilities_generalized}
\end{figure*}

The results of these larger-scale simulations are presented in Fig.~\ref{fig:large_CL} and summarized in Table~\ref{table:large_CL}. For both the 16- and 32-qubit QCNN, the Hard strategy clearly demonstrates superior performance, reaching a higher test accuracy compared to both the Standard and Easy approaches. This is particularly remarkable in the 32-qubit case: while the absolute accuracy for all methods decreases, this drop is an expected consequence of the increased problem complexity coupled with the simplified, lower-parameter ansatz required to make the classical simulation of 32 qubits computationally feasible. Furthermore, within the restricted number of training epochs shown, this highlights a core advantage of our approach. In practical large-scale quantum machine learning, the computational budget for training steps is often strictly limited. The fact that the performance gap between the Hard strategy and the others remains (and even widens) under these strict constraints confirms that prioritizing difficult examples yields a more efficient optimization path. Thus, even if absolute performance is limited by the ansatz capacity and number of steps, the trainability benefits of data-ordering strategies persist.

Quantitatively, Table~\ref{table:large_CL} shows that the Hard strategy achieves a best test accuracy of $71.7\%$ for 16 qubits and $58.7\%$ for 32 qubits. This represents a performance improvement over the Standard approach of $4.3$ percentage points for the 16-qubit case and $5.5$ percentage points for the 32-qubit case. This widening gap suggests that the inductive bias introduced by focusing on challenging examples might be particularly relevant for navigating the more complex optimization landscapes of larger quantum systems. The Standard method, which samples uniformly, is more likely to be hindered by non-informative gradients, whereas the Hard strategy effectively steers the model toward learning the most decisive features of the data. These findings support the conclusion that a self-paced, difficulty-based curriculum is a scalable and effective approach for improving the performance of quantum machine learning models.

\subsection{Accuracy at the cost of confidence}\label{sec:accuracy}

In this section, we try to understand further the differences between the Standard approach and the best-performing method within the proposed framework: the self-paced Hardest learning strategy. We analyze the probabilities for the different phases as determined by the 8-qubit QCNN along a specific cut in the phase diagram in Fig.~\ref{fig:diagram}.

We begin by considering the \emph{generalized cluster Hamiltonian}. We focus on a cut along the $j_2$ coupling constant, fixing $j_1=1$. The model undergoes three phase transitions: (1) from symmetry-protected topological to ferromagnetic at $j_2=2$, (2) from ferromagnetic to trivial at $j_2=0$, and (3) from trivial to symmetry-protected topological at $j_2=1$. Figure~\ref{fig:probabilities_generalized}(a) illustrates the probabilities for the different phases under both the Standard (top) and self-paced Hardest (bottom) strategies. The phase diagram is accurately captured by the QCNN trained with the self-paced strategy, whereas the Standard approach fails to precisely detect the phase boundaries.

Next, we examine the \emph{bond-alternating XXZ Hamiltonian}. Here, we zoom-in cut along the $j_1/j_2$ coupling constants, fixing $\delta=3$. The model features two phase transitions: (1) from trivial to symmetry-broken antiferromagnetic at $j_1/j_2\approx 1$, and (2) from symmetry-broken antiferromagnetic to topological at $j_1/j_2\approx 2$. Figure~\ref{fig:probabilities_generalized}(b) shows the probabilities for the different phases under both the Standard (top) and self-paced Hardest (bottom) strategies. Once again, the QCNN trained with the self-paced strategy accurately captures the phase diagram, while the Standard approach misclassifies a large fraction of the ground states.

Overall, the Hardest approach outperforms the Standard strategy in terms of accuracy, both within individual phases and the phase boundaries. However, it is worth noting that the Hardest approach exhibits less confidence in its predictions, as indicated by the smaller differences between the probabilities of different classes. This suggests a potential trade-off between accuracy and prediction confidence. Despite its superior performance, we still observe slight deviations in the predicted transition points from the actual phase boundary. This discrepancy is expected due to the relatively large correlation length near the phase boundary, which exceeds what the capacity of the 8-qubit QCNN can capture accurately.

\section{Conclusions}\label{sec:discussion}

In this work, we have introduced a training approach for quantum machine learning models that leverages the structure of the training data itself to introduce additional inductive bias, thereby improving learning capabilities. Drawing inspiration from classical methods like curriculum learning and hard example mining, we developed a framework tailored to quantum models. We provided some analytical insights and extensive numerical evidence that demonstrates how strategically ordering the quantum data can significantly improve the performance of quantum models.

Analytically, we have established that data ordering can mitigate the effects of barren plateaus (Proposition~\ref{prop:barren_plateaus}) and tighten the upper bound on the expected empirical risk to enhance convergence (Proposition~\ref{prop:convergence}). Our numerical experiments have focused on the task of quantum phase recognition in quantum spin chains, explicitly validated these theoretical insights. In particular, the accelerated convergence and improved final accuracy observed when using data-ordering strategies directly reflect the analytical benefits of reduced gradient variance outlined in Proposition~\ref{prop:convergence}. Among all the data-ordering strategies, we have observed that prioritizing harder data points early in the training process leads to superior performance compared to traditional training methods. The self-paced learning strategy, which dynamically adjusts data presentation based on current performance, has proven to be particularly effective, consistently achieving higher accuracy than other strategies. Crucially, 
our large-scale simulations on systems of up to 32 qubits have confirmed that this performance advantage is maintained as system size increases. This emphasizes the potential of data ordering in quantum models, not just as a theoretical concept but as a relevant factor in practical applications.

While the framework presented here showcases 
interesting results, it also opens up several avenues for future research. One possible direction is to extend these methods to tasks beyond quantum phase recognition, such as unitary learning~\cite{cincio2018learning, beer2020training, cincio2021machine, caro2022generalization, yu2023optimal} or quantum error correction~\cite{cong2019quantum, locher2023quantum}. 
Furthermore, exploring different complexity measures---such as those based on 
notions of entanglement---or combining multiple complexity metrics, including those used in this work, could help refine the scoring and pacing functions. Integrating this data-centric approach 
with other warm-start techniques, such as those focused on parameter initialization, could further enhance model performance. On 
a higher level, this work further reinforces the importance of data-centric inductive biases in quantum machine learning models, suggesting that advancing this concept could lead to more sophisticated quantum circuit design strategies that can optimize performance from the outset, while also underscoring the value 
of leveraging established classical machine learning methods to foster meaningful progress in quantum machine learning.
\bigskip

\section*{Code and data availability}
    The code and data generated during the current study are available in Ref.~\cite{github_repository}.

\bigskip

\acknowledgments
The authors would like to thank Elies Gil-Fuster and Adrián Pérez-Salinas for useful feedback on a previous version of this manuscript. E.~R.-A.\ acknowledges funding from the Government of Spain (Severo Ochoa CEX2019-000910-S, FUNQIP, and European Union NextGenerationEU PRTR-C17.I1), Fundació Cellex, Fundació Mir-Puig, Generalitat de Catalunya (CERCA program) and European Union (PASQuanS2.1, 101113690). E.~R.-A.\ is a fellow of Eurecat's \enquote{Vicente López} PhD grant program. F.~J.~S., J.~E., and C.~B.-P.\  are supported by the BMFTR (MUNIQC-Atoms, Hybrid++, QoSol, PasQuops), 
the BMWK (EniQma), the Munich Quantum Valley (K-8), the Quantum Flagship (PasQuans2, Millenion), the Einstein Foundation (Einstein Research Unit on  Quantum Devices), Berlin Quantum, and the DFG (CRC 183, SPP 2514). \smallskip

\bibliography{citations.bib}

\vspace*{0.5cm}

\appendix

\setcounter{theorem}{0}

\section{Proof of Proposition 1}
\label{app:prop_1}

In this section, we re-state and prove Proposition~\ref{prop:barren_plateaus}.

\begin{proposition}[Mitigation of barren plateaus]
Let $\{\ket{\psi_i} \in \calH\}$ be 
a collection of quantum state vectors.
Let $f_{\vvartheta}: \mathcal{H} \rightarrow \mathcal{Y}$ be a parameterized quantum model with parameters $\vvartheta \in \Theta$ and $\calY \subseteq \bbR$ the label space. Consider a loss function $\ell(\vvartheta; (\ket{\psi}, y))$ associated with data points $\{(\ket{\psi_i}, y_i) \in \mathcal{H} \times \mathcal{Y}\}$. Suppose there exists a scoring function $s:\calH \to [0,1]$ such that the expected squared norm of the gradient
    \begin{equation}
        G(z) = \mathbb{E}\left[\norm{\nabla_{\vvartheta} \ell(\vvartheta; (\ket{\psi}, y)}^2 \mid s(\psi) = z\right]
    \end{equation}
    is a non-increasing function of $s$. Then, this learning setup defined by $s$ increases the expected squared gradient magnitude during early training epochs compared to standard training (random data order).
\end{proposition}
\begin{proof}
    Let $F$
    be the cumulative distribution function of the scores over the data set given by
    \begin{equation}
        F(z) = \mathbb{P}\left[s(\psi_i) \leq z\right].
    \end{equation}
    The learner has access to the subset $\mathcal{D}_t = \{(\ket{\psi_i}, y_i) \mid s(\psi_i) \leq z_t\}$. The expected gradient magnitude at epoch $t$ is defined as
    \begin{equation}
        \mathbb{E}_{\mathcal{D}_t} \left[\norm{g}^2\right] = \frac{1}{F(z_t)} \int_0^{z_t} G(z) dF(z) \,,
    \end{equation}
    where $g$ is the gradient of the loss function with respect to the parameters $\vvartheta$. For random data presentation, all data points are equally likely to be sampled, regardless of their scores. Therefore,
    \begin{equation}
        \mathbb{E}_{\text{rand}} \left[\norm{g}^2\right] = \int_0^1 G(z) dF(z) \,.
    \end{equation}
    Since $G$
    is non-increasing, for $z \leq z_t$, $G(z) \geq G(z_t)$, and similarly, for $z \geq z_t$, $G(z) \leq G(z_t)$. Therefore, we can lower bound
    \begin{equation}
        \mathbb{E}_{\mathcal{D}_t} \left[\norm{g}^2\right] \geq \frac{1}{F(z_t)} G(z_t) \int_0^{z_t} dF(z) = G(z_t) \,,
    \end{equation}
    and upper bound
    \begin{equation}
        \mathbb{E}_{\text{rand}} \left[\norm{g}^2\right] \leq G(z_t) F(z_t) + G(z_t)(1-F(z_t)) = G(z_t) \,,
    \end{equation}
    which leads to
    \begin{equation}
        \mathbb{E}_{\mathcal{D}_t} \left[\norm{g}^2\right] \geq \mathbb{E}_{\text{rand}} \left[\norm{g}^2\right],
    \end{equation}
    which proves the statement to be shown.
\end{proof}

\section{Proof of Proposition 2}\label{app:prop_2}
In this section, we re-state and prove Proposition~\ref{prop:convergence}.

\begin{proposition}[Faster convergence]
Let $\{\ket{\psi_i} \}\in \calH$ be 
a collection of
quantum state vectors. Let $f_{\vvartheta}: \mathcal{H} \rightarrow \mathcal{Y}$ be a parameterized quantum model with parameters $\vvartheta \in \Theta$ and $\calY \subseteq \bbR$ the label space. Consider a convex loss function with Lipschitz 
continuous gradients $\ell(\vvartheta; (\ket{\psi}, y))$ associated with data points $\{(\ket{\psi_i}, y_i) \in \mathcal{H} \times \mathcal{Y}\}$. Suppose there exists a scoring function $s:\calH \to [0,1]$ such that the scoring function $s(\ket{\psi})$ prioritizes data points with lower gradient variance at each epoch, and the pacing function $p$ increases monotonically over epochs $t = 1$ to $T$. Then, this learning setup defined by s and p achieves a theoretical upper bound on its expected empirical risk after T epochs that is less than or equal to the bound for standard training (uniform access to the entire dataset).
\end{proposition}

\begin{proof}
    In standard training, data points are selected uniformly at random. The variance of the stochastic gradient is
    \begin{equation}
        \sigma_{\text{rand}}^2 = \mathbb{E}_{(\ket{\psi}, y)\sim D} \left[\norm{\nabla \ell(f_{\vvartheta_t}; (\ket{\psi}, y))- \nabla R(\vvartheta_t)}^2\right]
    \end{equation}
    where $D$ is the uniform distribution over the entire dataset. In contrast, for the learning setup defined by $s$ and $p$, the variance of the stochastic gradient is
    \begin{equation}
        \sigma_{s,p}^2 = \mathbb{E}_{(\ket{\psi}, y)\sim D_t} \left[\norm{ \nabla \ell(f_{\vvartheta_t}; (\ket{\psi}, y))- \nabla R(\vvartheta_t)}^2\right]
    \end{equation}
    where 
   \begin{equation} \mathcal{D}_t = \{(\ket{\psi_i}, y_i) \mid s(\psi_i) \leq z_t\}.
   \end{equation}
   For convex functions with Lipschitz continuous gradients, the expected suboptimality after $T$ epochs is bounded by~\cite{shalev2014understanding}
    \begin{equation}
        \mathbb{E}\left[R(\vvartheta_T) - R(\vvartheta^\ast)\right] \leq \frac{\norm{\vvartheta_0 - \vvartheta^\ast}^2}{2\eta T} + \frac{\eta}{2T}\sum_{t=1}^T \sigma^2(t)\,,
    \end{equation}
    where $\vvartheta_0$ are the initial parameters, $\vvartheta^\ast$ are the optimal parameters minimizing $R(\vvartheta)$, and $\eta$ is the learning rate (which we assume to be constant for simplicity). Then, for the standard training we have
    \begin{equation}\label{eq:ub_rand}
        \mathbb{E}\left[R(\vvartheta_T^{\text{rand}}) - R(\vvartheta^\ast)\right] \leq \frac{\norm{\vvartheta_0 - \vvartheta^\ast}^2}{2\eta T} + \frac{\eta}{2T}\sum_{t=1}^T \sigma_{\text{rand}}^2(t)\,,
    \end{equation}
    and for the setup defined by $s$ and $t$
    \begin{equation}\label{eq:ub_sp}
        \mathbb{E}\left[R(\vvartheta_T^{s,p}) - R(\vvartheta^\ast)\right] \leq \frac{\norm{\vvartheta_0 - \vvartheta^\ast}^2}{2\eta T} + \frac{\eta}{2T}\sum_{t=1}^T \sigma_{s,p}^2(t).
    \end{equation}
    Given that for all $t$, $\sigma_{s,p}^2(t) \leq \sigma_{\text{rand}}^2(t)$, it follows that
    \begin{equation}\label{eq:ub_comparison}
        \sum_{t=1}^T \sigma_{s,p}^2(t) \leq \sum_{t=1}^T \sigma_{\text{rand}}^2(t)\,.
    \end{equation}
    Let $K_{\text{rand}}$ and $K_{s,p}$ denote the right-hand sides of the inequalities in \eqref{eq:ub_rand} and \eqref{eq:ub_sp}, respectively. These are the upper bounds on the expected sub-optimality for our methods and standard training. From \eqref{eq:ub_comparison}, it is clear that the second term in $K_{s,p}$ is less than or equal to the second term in $K_{\text{rand}}$. Therefore, we can conclude that the overall upper bound for our method is tighter:
    \begin{equation}
        K_{s,p} \leq K_{\text{rand}}\,.
    \end{equation}    
    This proves that our learning setup achieves an equal or tighter theoretical upper bound on the expected empirical risk compared to standard training
\end{proof}

\section{Additional numerical details}\label{app:numerics}
This appendix provides a detailed description of the numerical methods, model architecture, and hyperparameters used for the quantum phase classification experiments on the generalized cluster Hamiltonian (GCH) and the bond-alternating XXZ Hamiltonian (XXZ).

\subsection{Dataset generation}
The training and test datasets were generated using a balanced sampling procedure to ensure an equal representation of each quantum phase. For each model, we sample parameters points uniformly at random from the regions of the 2D phase diagram (see Fig.~\ref{fig:diagram}) corresponding to each distinct phase. The training set was constructed from $50$ total parameter points, resulting in $12$-$13$ samples per phase for the GCH (four phases) and $16$-$17$ samples per phase for the XXZ model (three phases). The test set was generated using the same balanced distribution but comprised a larger set of $1000$ samples ($250$ for the large-scale simulations) to ensure a robust evaluation. The final datasets consist of the ground states corresponding to these parameter points.

\subsection{Quantum circuit architecture}
We employed a Quantum Convolutional Neural Network (QCNN) architecture with periodic boundary conditions. The circuit consists of alternating convolutional and pooling layers, where all gates within a given layer share the same trainable parameters. The final two qubits are measured in the computational basis to produce a probability vector over the four outcomes $\{00, 01, 10, 11\}$, which serves as the model's output for classification.

The specific implementation of the QCNN, including the number of qubits, layers, and the type of variational gates, was adapted for each set of experiments. A comprehensive summary of these configurations is provided in Table \ref{table:experiment_conf}. For most experiments, the convolutional and pooling layers were constructed from general two-qubit unitary gates. However, for the physics-inspired learning experiments (Sec.~\ref{sec:physics-inspired_CL}), the computational cost of calculating the $\mathfrak{g}$-purity necessitated the use of a more restricted ansatz based on matchgate circuits. Similarly, for the large-scale simulations, we employed a simplified gate ansatz with fewer parameters due to computational constraints. In particular, we utilized the matrix product state backend from \texttt{TensorCircuit}~\cite{zhang2023tensorcircuit} to simulate the large-scale quantum circuits. A bond dimension of $\chi = 20$ was employed
for the simulations of 16- and 32-qubit QCNNs.

\begin{figure}[h!]
    	\centering
    	\includegraphics[width=0.49\textwidth]{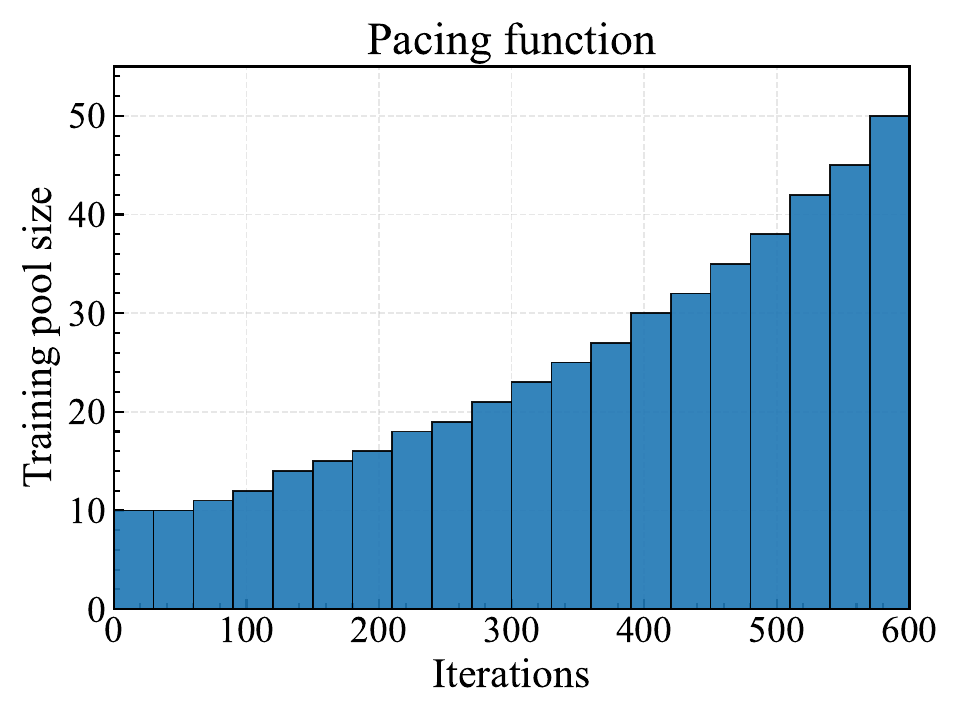}
    	\caption{\justifying
        The exponential pacing function used for a $600$ iterations training run. The function dictates the size of the data pool from which mini-batches are sampled at each epoch.}
    	\label{fig:pacing_function}
\end{figure}

\subsection{Training hyperparameters}
The model parameters were optimized using the Adam algorithm with a learning rate of $0.01$ and hyperparameters $\beta_1=0.9$, $\beta_2=0.999$, and $\epsilon=10^{-8}$. The parameters were initialized by sampling from a Gaussian distribution with a mean of $0$ and a standard deviation of $1/\sqrt{n}$, where $n$ is the number of qubits. All reported results are averaged over $10$ independent runs ($6$ for the large-scale simulations), each with a different random data sample and parameter initialization, to ensure statistical robustness.

\subsection{Curriculum learning implementation}
The curriculum is defined by a pacing function, which controls the size of the data pool accessible to the learner at each training epoch. We experimented with linear, logarithmic, and exponential pacing functions. The exponential function, illustrated in Figure~\ref{fig:pacing_function} as an example, consistently yielded the best performance and was therefore used for all reported curriculum learning results.

\begin{table*}[]
\centering
\begin{tabular}{|c|c|c|c|c|c|}
\hline
\textbf{Experiment} & \textbf{Qubits} & \textbf{2-qubit gate} & \textbf{Gate params.} & \textbf{Layers} & \textbf{Total params.} \\
\hline
Self-taught & 8 & General gates & 15 & 3 & 75 \\
\hline
Self-paced & 8 & General gates & 15 & 3 & 75 \\
\hline
Physics insp. & 8 & Match gates & 5 & 3 & 25 \\
\hline
Large-scale & 16 & Simplified & 8 & 4 & 61 \\
\hline
Large-scale & 32 & Simplified & 8 & 5 & 77 \\
\hline
\end{tabular}
\caption{Summary of the QCNN architecture and parameters for each numerical experiment.} \label{table:experiment_conf}
\end{table*}

The curriculum starts with an accessible pool of $10$ samples, which grows exponentially to include all $50$ training samples by the end of the training. It is important to note that each epoch, a mini-batch of $10$ samples was drawn uniformly at random from the \textit{currently accessible data pool}, not from the entire dataset.

\section{Complexity analysis and overhead of the self-paced strategy}\label{app:self-paced}
The self-paced learning strategy requires continuously re-scoring and re-ordering the training data based on the model's evolving hypothesis. While dynamically evaluating the loss over the entire dataset introduces a computational overhead, this cost must be directly contextualized against the dominant resource requirements of gradient estimation when executing on physical quantum hardware.

For a dataset of size $N$, scoring the data requires $N$ forward passes through the quantum circuit. Subsequently, sorting the scores incurs a classical computational cost of $\mathcal{O}(N\, \mathrm{log} N)$. In contrast, the optimization process on quantum hardware relies heavily on the parameter-shift rule for exact gradient calculation~\cite{parameterrule2019, wierichs2022general}. For a parameterized quantum circuit with $P$ trainable parameters and a mini-batch of size $L$, computing the gradient for a single optimization step requires $\mathcal{O}(P \times L)$ circuit executions. 

In our numerical implementations, re-scoring was performed at every optimization step. Even under this frequent update schedule, the $\mathcal{O}(P \times L)$ cost of gradient estimation vastly dominates the $N$ forward passes required for scoring, as $P$ is typically large and grows as models scale to achieve better expressivity. Furthermore, our framework is designed to support periodic re-scoring to optimize performance. By updating scores only every few iterations, the additional circuit execution budget can be reduced to a negligible fraction of the total training cost.

Finally, the classical sorting overhead of $\mathcal{O}(N\, \mathrm{log} N)$ remains negligible compared to the time required for quantum circuit executions. Therefore, the overhead introduced by the self-paced strategy represents only a small fraction of the total execution budget. By significantly improving the trainability and convergence rate of the model within a fixed number of steps, this method effectively amortizes its own scoring cost, proving to be a highly resource-efficient strategy for near-term quantum devices.

\section{Extended number of iterations in the self-taught strategy}\label{app:extended_nums}
To observe the asymptotic behavior of the self-taught methods, we extended the training to 1200 epochs for both the generalized cluster and the bond-alternating XXZ models on ten different independent runs. Figure~\ref{fig:self-taught_extended_it} illustrates the test accuracy over this extended training window. The vertical dashed line at epoch $600$ marks the end of the pacing function schedule.

\begin{figure}[t!]
\centering
             \includegraphics[width=0.48\textwidth]{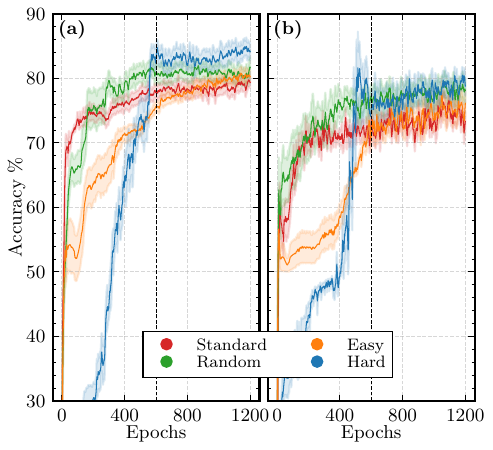}
             \caption{\justifying
             \label{fig:self-taught_extended_it}Average accuracy on the test set as a function of the training epochs (extended to 1200 epochs) for the quantum phase recognition task for the \textbf{(a)} \emph{generalized cluster model} and \textbf{(b)} the \emph{bond-alternating XXZ model}, using a self-taught strategy. The vertical dashed line at epoch $600$ indicates the completion of the pacing function schedule, after which the model is trained on the full dataset. The shaded areas represent the standard error of the mean over ten runs with different random initialization and training data.}

    \vspace{1.5em}

    \begin{minipage}{\columnwidth}
        \centering
        \renewcommand{\arraystretch}{1.2}

        \begin{tabular}{|Sc|Sc|Sc||Sc|Sc|} \cline{2-5}
\multicolumn{1}{c|}{} & \multicolumn{2}{c||}{\textbf{Generalized cluster}} & \multicolumn{2}{c|}{\textbf{Bond-alternating XXZ}} \\ \cline{2-5}
\multicolumn{1}{c|}{} & \multicolumn{1}{c|}{\,\,Training $(\%)$\,\,} & \multicolumn{1}{c||}{\,\,Test $(\%)$\,\,} & \multicolumn{1}{c|}{\,\,Training $(\%)$\,\,} & \multicolumn{1}{c|}{\,\,Test $(\%)$\,\,} \\ \hline
\,\textit{Standard}\, & 79.2 & 76.6 & 79.0 & 76.6 \\ \hline
\,\textit{Random}\, & 82.8 & 81.9 & 81.4 & 79.5\\ \hline
\,\textit{Easy}\, & 81.8 & 80.8 & 78.4 & 77.7\\ \hline
\,\textit{Hard}\, & \textbf{85.9} & \textbf{85.0} & \textbf{82.1} & \textbf{81.7}\\
\hline\end{tabular}

        \captionsetup{justification=justified,singlelinecheck=false}
        \captionof{table}{\label{table:ST_CL_extended_it}
        \justifying
        Average best accuracy over 
ten runs for the quantum phase recognition task on the \emph{generalized cluster model} and the \emph{bond-alternating XXZ model}, using a self-taught strategy (extended to $1200$ epochs). The best-performing method prioritizes hard examples early in the training.}
    \end{minipage}
\end{figure}

It is important to emphasize that the core objective of these data-ordering strategies is not necessarily to reach asymptotic saturation over an arbitrarily large number of epochs. Rather, the primary goal is to improve model performance and accelerate convergence within a fixed, and often strictly limited, number of training steps. This is particularly critical in the context of near-term quantum hardware, where circuit executions are expensive and optimization budgets are highly constrained. Nevertheless, analyzing the performance after the completion of the pacing function schedules provides valuable insights into the stability of the optimization trajectory.

As shown in Fig.~\ref{fig:self-taught_extended_it}(a) for the generalized cluster model, the Hard strategy successfully steers the optimization process away from poor local minima and toward a more favorable region of the parameter space early in the training. Once the model is situated within this region, subsequent training on the complete dataset preserves this initial advantage, resulting in a higher stabilized final accuracy compared to the other orderings.

For the bond-alternating XXZ model in Fig.~\ref{fig:self-taught_extended_it}(b), the extended training clarifies the sharp accuracy drop observed near epoch $600$. This decrease in the final optimization steps directly coincides with the introduction of the easiest, least informative examples into the training batches. Intriguingly, this suggests that in certain quantum machine learning scenarios, exposure to non-informative data points might actively degrade the model's performance, and it might be beneficial to permanently remove the least informative samples to improve overall accuracy. Following this transient drop, the Hard strategy stabilizes alongside the other methods. Table~\ref{table:ST_CL_extended_it} summarizes the best test and training accuracies for all methods evaluated.

\end{document}